\documentclass[12pt,reqno]{amsart}
\usepackage{amssymb}
\usepackage{amsmath}

\usepackage{amsmath,amssymb}
\usepackage[dvips]{graphicx,psfrag}

\newcommand{\rank}{\mbox{rank}}

\newcommand{\rmv}[1]{}

\newcommand{\mH}{\mathcal{H}}

\newtheorem{thm}{Theorem}
\newtheorem{lem}[thm]{Lemma}

\newtheorem{cor}[thm]{Corollary}

\newtheorem{asp}[thm]{Assumption}

\newtheorem{cnj}[thm]{Conjecture}

\begin{document}

\title[Solving the 100 Swiss Francs Problem]{Solving the 100 Swiss Francs Problem}

\author[M.Zhu] {Mingfu Zhu}
\address{Center for Human Genome Variation\\
Duke University\\
Durham, NC 27708, USA}

\email{mingfu.zhu@duke.edu}

\author[G.Jiang]{Guangran Jiang}
\address{Department of Computer Science\\
Zhejiang University\\
Hangzhou, Zhejiang, 310027, China}

\email{rpggpr@zju.edu.cn}

\author[S.Gao]{Shuhong Gao}
\address{Department of Mathematical Sciences\\
Clemson University\\
Clemson, SC 29634-0975, USA}

\email{sgao@clemson.edu}



\begin{abstract}
Sturmfels offered 100 Swiss Francs in 2005 to a conjecture, which deals with a special case of the maximum likelihood estimation for a latent class model. This paper confirms the conjecture positively.
\end{abstract}
\maketitle

\section{The conjecture and its statistical background}
Sturmfels \cite{S2} proposed the following problem: Maximize the
likelihood function
\begin{equation} \label{eq1.1}
L(P)=\prod_{i=1}^4 p_{ii}^4 \times \prod_{i \not= j}p_{ij}^2
\end{equation}
over the set of all $4 \times 4$-matrices $P=(p_{ij})$ whose entries
are nonnegative and sum to $1$ and whose rank is at most two. Based
on numerical experiments by employing an expectation-maximization(EM) algorithm, Sturmfels \cite{PS,S2}
conjectured that the matrix
 $$P=\frac{1}{40}
\begin{pmatrix}
  3 & 3 & 2 & 2 \\
  3 & 3 & 2 & 2 \\
  2 & 2 & 3 & 3 \\
  2 & 2 & 3 & 3 \\
\end{pmatrix}$$
is a global maximum of $L(P)$. He offered 100 Swiss francs for a
rigorous proof in a postgraduate course held at ETH Z\"{u}rich in
2005.

Partial results were given in the paper in \cite{FHRZ}, where the general
statistical background for this problem is also presented. This
problem is a special case of the maximum likelihood estimation for a
latent class model. More precisely, by following \cite{FHRZ}, let $(X_1,
\ldots, X_d)$ be a discrete multivariate random vector where each
$X_j$ takes value from a finite state set $S_j=\{1,\dots,s_{j}\}$.
Let $\Omega=S_1\times \cdots \times S_d$ be the sample space. For each
$(x_1,\ldots,x_d) \in \Omega$, the joint probability mass function
of $(X_1, \ldots, X_d)$ is denoted as
$$p(x_1,\ldots,x_d)=P\{(X_1,\ldots,X_d)=(x_1,\ldots,x_d)\}.$$
The variables $X_1, \ldots, X_d$ may not be mutually independent
generally. By introducing an unobservable variable $H$ defined on
the set $[r]=\{1,\ldots r\}$, $X_1, \ldots, X_d$ become mutually
independent. The joint probability mass function in the newly formed
model is
\begin{eqnarray*}
p(x_1,\ldots,x_d,h)&=&P\{(X_1,\ldots,X_d,H)=(x_1,\ldots,x_d,h)\}\\
&=&p(x_1|h)\cdots p(x_d|h) \lambda_h
\end{eqnarray*} where
$\lambda_h$ is the marginal probability of $P\{H=h\}$ and $p(x_j|h)$
is the conditional probability $P\{X_j=x_j|H=h\}$. We denote this new
$r$-class mixture model by $\mH$. The marginal distribution of
$(X_1,\ldots,X_d)$ in $\mH$ is given by the probability mass
function (which is also called \textit{accounting equations} \cite{HL})
$$\label{accounting}
p(x_1,\ldots,x_d)=\sum_{h\in [r]}p(x_1,\ldots,x_d,h)=\sum_{h\in
[r]}p(x_1|h)\cdots p(x_d|h) \lambda_h.$$

In practice, a collection of samples from $\Omega$ are observed. For
each $(x_1,\ldots,x_d)$, let $n(x_1,\ldots,x_d)\in \mathbb{N}$ be
the number of observed occurrences of $(x_1,\ldots,x_d)$ in the
samples. While the parameters $p(x_1|h)$,$\cdots$, $p(x_d|h)$,
$\lambda_h, p(x_1,\ldots,x_d)$ are unknown. The maximum likelihood
estimation problem is to find the model parameters that can best
explain the observed data, that is, to determine the global maxima
of the likelihood function
\begin{equation}\label{mle}
L(\mH)=\prod_{(x_1,\ldots,x_d) \in \Omega}
p(x_1,\ldots,x_d)^{n(x_1,\ldots,x_d)}.\nonumber
\end{equation}

Since each $p(x_1,\ldots,x_d)$ is nonnegative, it is equivalent but more convenient to use
the log-likelihood function
\begin{equation}\label{mle2}
l(\mH)=\sum_{(x_1,\ldots,x_d) \in \Omega} n(x_1,\ldots,x_d)\ln p(x_1,\ldots,x_d),
\end{equation}
where we define $\ln(0)=-\infty$. Finding the maxima of (\ref{mle2})
is difficult and remains infeasible by current symbolic software \cite{CHKS,HKS}.
We can only handle some special cases: small models or highly
symmetric table. The 100 Swiss francs problem is the special case of
$\mH$ when $d=2$, $S_1=S_2=\{$A,C,G,T$\}$, $s_1=s_2=4$ and $r=2$. It
is related to a DNA sequence alignment problem as described in \cite{PS}.
In that example, the contingency table for the observed counts of
ordered pairs of nucleotides (i.e.\ AA, AC, AG, AT, CA, CC,
$\cdots$) is
$$\begin{tabular}{cc}
&\hspace{-0.2cm}A\hspace{0.1cm} C\hspace{0.1cm} G\hspace{0.1cm} T\hspace{0.1cm}\ \\
A&\\
C&\\
G&\\
T&\raisebox{0.65cm}[0pt]{\hspace{-0.3cm}$\begin{pmatrix}
  4 & 2 & 2 & 2 \\
  2 & 4 & 2 & 2 \\
  2 & 2 & 4 & 2 \\
  2 & 2 & 2 & 4 \\
\end{pmatrix}$}\\
\end{tabular}\hspace{-0.2cm}.$$
So the likelihood function (\ref{mle}) in this example is exactly
(\ref{eq1.1}).

Even for this simple case, the problem is surprisingly difficult. We know that the global maxima must exist, as the region of the parameters is compact.
By using an EM algorithm
or Newton-Raphson method and starting from suitable initial points,
one can find some local maxima of the likelihood function.
However, the global maximum property is not guaranteed. We prove that Sturmfels' conjectured solution is indeed a global maximum.

Our paper is organized as follows. We first derive some general
properties for optimal solutions in Section \ref{genprop}, then provide a theoretical solution to the conjecture in Sections \ref{handwr}. In \ref{comp}, we make some comments about using Gr\"obner basis technique in solving this problem and provide a computational solution. Lastly, we suggest several new conjectures in more general cases.

\section{Proof of the conjecture}
\subsection{General Properties}\label{genprop}

We focus on general $n\times n$ matrices $P=(p_{ij})$ in this
section. For convenience we scale each entry of $P$ by $n^2$ so the
entries sum to $n^2$, and take square root of the original
likelihood function. So we may assume that
\begin{equation} \label{product}
L(P)=\prod_{i=1}^n p_{ii}^2 \times \prod_{i \not= j}p_{ij}.
\end{equation}
The problem is
\begin{center}
\begin{tabular}{ll}
Maximize:   & $L(P)$ \\
Subject to:  &  $\sum \limits_{1\leq i, j \leq n} p_{ij} = n^2$,  and \vspace{0.1cm}\\
                  & $p_{ij} \geq 0$, $1\leq i, j \leq n$.
\end{tabular}
\end{center}
Suppose $P=(p_{ij})_{n\times n}$ is a global maximum of $L(P)$.
It is easy to see that  $P$ cannot be the following $n\times n$ matrix
$$J=\begin{pmatrix}
      1 & \ldots & 1 \\
      \vdots &   & \vdots \\
      1 & \ldots & 1 \\
    \end{pmatrix}.
$$

Since the function (\ref{product}) is a continuous function in
$p_{ij}$'s, if one of the entries of $P$ approaches 0, the product
has to approach 0 too, as the other entries are bounded by $n^2$.
Hence the optimal solutions must occur in interior points and we
don't need to worry about the boundary where some $p_{ij}=0$.

Therefore, in the subsequent discussion, we may  assume that $P \neq J$ and all its entries are positive.
We show that $P$ must have certain symmetry properties.

\begin{lem} \label{rowsum}
For an optimal solution $P$, its row sums and column sums must all
equal $n$.
\end{lem}
\begin{proof}
Let $\sum \limits_{j=1}^n p_{ij}=s_i$. Then $\sum \limits_{i=1}^n
s_i=n^2$ and $\prod \limits_i s_i \le n^n$ with equality if and only
if $s_i=n$ for all $i$. Let $\bar{p}_{ij}=\frac{n}{s_i} p_{ij}$ and
$\bar{P}=(\bar{p}_{ij})_{n\times n}$. Then $\rank(\bar{P})=\rank(P)$
and $\sum \limits_{i,j} \bar{p}_{ij}=n^2$. However,
$$L(\bar{P})= L(P) \cdot \left(\frac{n^n}{\prod \limits_i s_i}\right)^{n+1}  \ge L(P)$$
 with equality if and only if
$s_i=n$ for all $i$. Since $P$ is a global maximum, $L(\bar{P}) \le
L(P)$. Therefore each row sum equals $n$.  Similarly, each column
sum equals $n$ as well.
\end{proof}

We shall express $P$ in a form that involves fewer variables and has
no rank constraint. Since $P$ has rank at most two, by singular
value decomposition theorem, there are column vectors $u_{1}, u_{2},
v_{1}$ and $v_{2}$ of length $n$ such that
\[P=\sigma_{1} u_{1}v_{1}^{t} +\sigma_{2} u_{2}v_{2}^{t}\]
for some nonnegative numbers $\sigma_{1}$ and $\sigma_{2}$. By Proposition \ref{rowsum}, $P$ has equal
row and column sums, so $P$ has the vectors $(1, 1, \ldots, 1)$ and $(1, 1, \ldots, 1)^{t}$ as its left and right
eigenvectors both with eigenvalue $1$. Hence we may assume that $\sigma_{1}=1$ and $u_{1}= v_{1}= (1, 1, \ldots, 1)^{t}$.
Let $v_{2}=(a_{1}, a_{2}, \ldots, a_{n})^{t}$ and $\sigma_{2 }u_{2}=(b_{1}, b_{2}, \ldots, b_{n})^{t}$.
Then $P$ has the form
\begin{equation} \nonumber
 P=
J+\begin{pmatrix}
b_1\\
\vdots\\
b_n\\
\end{pmatrix}
    (a_{1}, a_{2}, \ldots, a_{n})\\
=\begin{pmatrix}
    1+a_1 b_1 & \cdots & 1+a_n b_1 \\
    \vdots &  1+a_i b_j & \vdots \\
    1+a_1 b_n & \cdots & 1+a_n b_n \\
  \end{pmatrix}.
\end{equation}
In this form, $P$ has rank at most two. Also, the condition $\sum
\limits_{ij}p_{ij}=n^2$ becomes
\begin{equation}\label{aibizero}
\sum \limits_{i=1}^n a_i \cdot \sum \limits_{i=1}^n b_i =0.
\end{equation}

We have transformed the original problem to the following
optimization problem:
\begin{center}
\begin{tabular}{ll}
Maximize:   &  $l(P)=2\sum \limits_{i=1}^n \ln(1+a_ib_i)  + \sum \limits_{i \not=j}\ln(1+a_ib_j)$ \\
Subject to:  & Equation (\ref{aibizero}) and  $1 + a_{i}b_{j} > 0$,
$1\leq i, j \leq n$.
\end{tabular}
\end{center}
The Lagrangian function would be
\begin{equation} \nonumber \label{lf}
\Lambda(P,\lambda)=l(P)+\lambda \sum \limits_{i=1}^n a_i \cdot \sum
\limits_{i=1}^n b_i
\end{equation}
where $\lambda \in \mathbb{R}$. Any local extrema must satisfy
\begin{equation} \label{diffa}
\frac{\partial \Lambda(P,\lambda)}{\partial a_i}=\sum_{j=1}^n
\frac{b_j}{1+a_i b_j}+\frac{b_i}{1+a_i b_i}+\lambda \sum_{j=1}^n
b_j=0, \ \ 1 \leq i \leq n,
\end{equation}
and
\begin{equation} \label{diffb}
\frac{\partial \Lambda(P,\lambda)}{\partial b_j}=\sum_{i=1}^n
\frac{a_i}{1+a_i b_j}+\frac{a_j}{1+a_j b_j}+\lambda \sum_{i=1}^n
a_i=0, \ \ 1 \leq j \leq n.
\end{equation}

By Lemma \ref{rowsum}, for an optimal solution $P$, its row sums and
column sums must be all equal to $n$. This means that
\begin{equation} \label{zero1}
\sum_{i=1}^n a_i=0,
\end{equation}
and
\begin{equation} \label{zero2}
\sum_{i=1}^n b_i=0.
\end{equation}
Plugging (\ref{zero1}) and (\ref{zero2}) into (\ref{diffa}) and
(\ref{diffb}) respectively, we obtain the following lemma.
\begin{lem} \label{diff} A global maximum $P$ must satisfy
\begin{equation} \label{diff1a}
 \sum_{j=1}^n \frac{b_j}{1+a_i b_j}+\frac{b_i}{1+a_i b_i}=0, \ \ 1
\leq i \leq n,
\end{equation}
and
\begin{equation} \label{diff1b}
\sum_{i=1}^n \frac{a_i}{1+a_i b_j}+\frac{a_j}{1+a_j b_j}=0, \ \ 1
\leq j\leq n.
\end{equation}
\end{lem}

Doing some simple algebra yields
\begin{cor} \label{diffcor} An optimal solution must satisfy
\begin{equation} \label{diff2a}
\sum_{j=1}^n \frac{1}{1+a_i b_j}+\frac{1}{1+a_i b_i}=n+1, \ \ 1 \leq
i \leq n,
\end{equation}
and
\begin{equation} \label{diff2b} \sum_{i=1}^n \frac{1}{1+a_i
b_j}+\frac{1}{1+a_j b_j}=n+1,  \ \ 1 \leq j\leq n.
\end{equation}
\end{cor}

\begin{proof} Multiply (\ref{diff1a}) by $a_i$ and then add
$\sum \limits_{j=1}^n \frac{1}{1+a_i b_j}+\frac{1}{1+a_i b_i}$ to
both sides, we can get (\ref{diff2a}).
\end{proof}

The $2n$ equations derived by clearing denominators of the equations
in Lemma \ref{diff} or Corollary \ref{diffcor} along with equations
(\ref{zero1}) and (\ref{zero2}) form a system of $2n+2$ polynomial equations with $2n$ unknowns,
whose solutions contain all global maxima. From computational  point of view,  we may
find all the solutions to this system of equations, say utilizing Gr\"{o}bner basis method,
and then pick a global maximum.
At the time we submitted this paper (in 2008), we could not solve the system for $n=4$ using Maple
on a computer with moderate computation power.
With both the advance in computer hardware and efficient implementations of
algorithms for computing Gr\"{o}bner basis, the system for $n=4$ now became solvable.
A computational solution for this problem is attached in Section \ref{comp}.
However, the system for $n=5$ remains unsolvable using our computers.

Our strategy below is to prove that $P$ should have high symmetry.
Firstly $a_i$'s and $b_i$'s are in the same order: if $a_i>a_j>0$,
then $b_i>b_j>0$ correspondingly (Lemma \ref{sign} and \ref{order}).
For the case $n=4$ once we force $a_1=b_1$ by scaling, we can
eventually prove $a_i=b_i$ for all other $i$'s (Lemma \ref{a2} and
\ref{a3a4}). With four $a_i$'s remained, we prove that the $a_i$'s
with the same signs must be identical. Finally one can solve the
system by hand. Note that Fienberg et. al. \cite{FHRZ} derived results similar to
Lemmas \ref{sign} and \ref{order}, but our approaches are simpler and completely different.

\begin{lem} \label{sign} For every $i$,
\begin{enumerate}
\item $a_i=0$ if and only if $b_i=0$, and
\item $a_i>0$ if and only if $b_i>0$.
\end{enumerate}
\end{lem}
\begin{proof} For the first part,  plugging in $a_i=0$ to the equation (\ref{diff1a}),
we have $\sum\limits_{j=1}^n b_j+b_i=0$, thus $b_i=0$. Similarly, if
$b_i=0$ then $a_i=0$.

\noindent For the second part ,  note that $g(x)=\frac{1}{x}$ is concave up in $(0,\infty)$. By Jensen's Inequality,
$$\sum_{j=1}^n \frac{1}{n} \cdot \frac{1}{1+a_ib_j}\ge \frac{1}{\sum \limits _{j=1}^n \frac{1}{n} (1+a_ib_j)}=1.$$
That is,
\begin{equation*}
\sum \limits _{j=1}^n \frac{1}{1+a_ib_j} \ge n.
\end{equation*} Compare with equation
(\ref{diff2a}), we get
$$\frac{1}{1+a_i b_i} \le 1,$$
so $a_i b_i \ge 0$. We conclude that $a_i>0$ if and only if
$b_i>0$.
\end{proof}

\begin{lem} \label{order} For $i$ and $j$,
\begin{enumerate}
\item $a_i=a_j$ if and only if $b_i=b_j$, and
\item $a_i>a_j$ if and only if $b_i>b_j$.
\end{enumerate}
\end{lem}
\begin{proof}
\noindent For the first part, suppose $b_{i}=b_{j}$. Then, by (\ref{diff1b}),
$$\sum_{k=1}^n \frac{a_k}{1+a_k b_i}+\frac{a_i}{1+a_i b_i}=0 \textrm{ and } \sum_{k=1}^n \frac{a_k}{1+a_k
b_j}+\frac{a_j}{1+a_j b_j}=0.$$ Then $\frac{a_i}{1+a_i
b_i}=\frac{a_j}{1+a_j b_j}$, so $a_{i}=a_{j}$. Then, using (\ref{diff1a}), we have $b_{i}=b_{j}$.

\noindent For the second part, switch $b_i,b_j$ in $P$ to form a new matrix
$\bar{P}$. Then we should have $L(P)\ge L(\bar{P})$ due to our assumption that $P$ is a global maximum. Note that
\begin{eqnarray*}
L(P)- L(\bar{P})
&=& C_{1}\cdot ((1+a_i b_i)^2 (1+a_i b_j) (1+a_j b_i) (1+a_j b_j)^2\\
&& -(1+a_i b_j)^2 (1+a_i b_i) (1+a_j b_j) (1+a_j b_i)^2) \\
&=& C_2 \cdot ((1+a_i b_i)(1+a_j b_j)-(1+a_i b_j)(1+a_j b_i))\\
&=& C_{2} \cdot (a_i b_i + a_j b_j-a_i b_j-a_j b_i)\\
&=& C_2  \cdot (a_i-a_j)(b_i-b_j)
\end{eqnarray*}
where $C_1,C_2$ are products of some entries of $P$, so $C_1, C_2$ are positive. Thus
$(a_i-a_j)(b_i-b_j)\ge 0$. Note that $a_i=a_j$ if and only if
$b_i=b_j$ by part(1), we conclude that $a_i>a_j$ if and only if
$b_i>b_j$. \qedhere
\end{proof}

\subsection{Theoretical solution}\label{handwr}

We complete the theoretical proof for the conjecture in this section. From now
on we focus on the case when $n=4$. By Lemma \ref{order}, we can
always assume $a_1 \ge a_2 \ge a_3 \ge a_4$ and $b_1 \ge b_2 \ge b_3
\ge b_4$. We know $a_1 \neq 0$, otherwise $b_1 = 0$ by Lemma \ref{sign}, hence $a_i=b_j=0$, which result in $P=J$. We also have $\frac{a_1}{b_1}>0$, so we can replace
$(a_1,a_2,a_3,a_4)$ in $P$ by
$\sqrt{\frac{a_1}{b_1}}(a_1,a_2,a_3,a_4)$ and
$(b_1,b_2,b_3,b_4)^{t}$ by
$\sqrt{\frac{b_1}{a_1}}(b_1,b_2,b_3,b_4)^{t}$. It turns out that $1+
\sqrt{\frac{a_1}{b_1}}a_i \sqrt{\frac{b_1}{a_1}}b_i=1+a_ib_j$ for
any $i$ and $j$, so we may always assume $a_1=b_1$. Thus $P$ can be
expressed as the form
\begin{equation} \label{a1b1}
\begin{pmatrix}
    1+a_1^2 & 1+a_2 a_1 & 1+a_3 a_1 & 1+a_4 a_1 \\
    1+a_1 b_2 & 1+a_2 b_2 & 1+a_3 b_2 & 1+a_4 b_2 \\
    1+a_1 b_3 & 1+a_2 b_3 & 1+a_3 b_3 & 1+a_4 b_3 \\
    1+a_1 b_4 & 1+a_2 b_4 & 1+a_3 b_4 & 1+a_4 b_4 \\
  \end{pmatrix}.
\end{equation}

If $a_2 \le 0$, we then replace $(a_1,a_2,a_3,a_4)$ in $P$ by
$(-a_4,-a_3,-a_2,-a_1)$ and $(b_1,b_2,b_3,b_4)^t$ by
$(-b_4,-b_3,-b_2,-b_1)^t$. The new matrix with $-a_4 \ge -a_3 \ge 0$
has the same likelihood function as $P$. Thus we may assume $a_1 \ge
a_2 \ge 0$. Without loss of generality, we may make the following
assumption.
\begin{asp} \label{asp} We can always assume the following
\begin{enumerate}
\item $a_1 \ge a_2 \ge a_3 \ge a_4$ and $b_1 \ge b_2 \ge b_3 \ge
b_4$,
\item $a_1=b_1>0$, and
\item $a_1 \ge a_2 \ge 0$.
\end{enumerate}
\end{asp}

The results in the rest of this section are all based on Assumption
\ref{asp}. Our first goal is to prove $a_2=b_2$.

\begin{lem} \label{a2} $a_2=b_2$.
\end{lem}
\begin{proof} If one of $a_2,b_2$ is 0, then $a_2=b_2=0$ by Lemma \ref{sign}.
We assume that both $a_2,b_2$ are nonzero.

Apply Corollary \ref{diffcor} to the first row of matrix
(\ref{a1b1}). We have
\begin{equation}\nonumber
\frac{2}{1+a_1^2}+\frac{1}{1+a_2 a_1}+\frac{1}{1+a_3
a_1}+\frac{1}{1+a_4 a_1}=5.
\end{equation}
Also
\begin{equation}\nonumber
a_1^2+a_2a_1+a_3a_1+a_4a_1=0.
\end{equation}
From the two equations above we get
\begin{equation} \label{eqn2}
a_3a_1 \cdot a_4a_1=f_1(a_1a_1,a_1a_2)
\end{equation}
where $f_1$ is a bivariate function in $x,y$ defined as
$$f_1(x,y)=\frac{2-x-y}{5-\frac{2}{1+x}-\frac{1}{1+y}}+x+y-1.
$$
Similarly, apply Corollary \ref{diffcor} to the second row of matrix
(\ref{a1b1}). We get
\begin{equation}\nonumber
\frac{1}{1+a_1 b_2}+\frac{2}{1+a_2 b_2}+\frac{1}{1+a_3
b_2}+\frac{1}{1+a_4 b_2}=5.
\end{equation}
Along with
\begin{equation}\nonumber
a_1b_2+a_2b_2+a_3b_2+a_4b_2=0,
\end{equation}
we get
\begin{equation} \label{eqn3}
a_3b_2 \cdot a_4b_2=f_1(a_2b_2,a_1b_2).
\end{equation}
Since $a_1,b_2$ are nonzero, we combine equations (\ref{eqn2}) and
(\ref{eqn3}) to get
\begin{equation} \label{eqn4}
\frac{f_1(a_1^2,a_1a_2)}{a_1^2}=\frac{f_1(a_2b_2,a_1b_2)}{b_2^2}.
\end{equation}
Normalizing (\ref{eqn4}) we can derive a trivariate polynomial
equation, say
\begin{equation} \label{eqn5}
f_2(a_1,a_2,b_2)=0.
\end{equation}

Symmetrically apply Corollary \ref{diffcor} to the first column and
the second column \ref{a1b1}, we get
\begin{equation}\label{eqn6}
\frac{f_1(a_1^2,a_1b_2)}{a_1^2}=\frac{f_1(a_2b_2,a_1a_2)}{a_2^2}.
\end{equation}
One can see that equation (\ref{eqn6}) is obtainable by switching
$a_2$ with $b_2$ in equation (\ref{eqn4}). Thus we have
\begin{equation} \label{eqn7}
f_2(a_1,b_2,a_2)=0.
\end{equation}
Subtracting (\ref{eqn7}) from (\ref{eqn5}) yields
\begin{equation}\nonumber\label{a2-b2}
f_2(a_1,a_2,b_2)-f_2(a_1,b_2,a_2)=0.
\end{equation}
Since we only switched $a_2$ and $b_2$ in polynomial $f_2$, there
must be a factor $a_2-b_2$ for $f_2(a_1,a_2,b_2)-f_2(a_1,b_2,a_2)$,
say
\begin{equation} \label{a2b2}
(a_2-b_2)f_3(a_1,a_2,b_2)=0,
\end{equation}
where
\begin{align*}
f_3(a_1,a_2,b_2)=&\ \ \ \ \ (20a_1^4b_2^2+15a_1^3b_2+3a_1^2b_2^2+2a_1b_2-4b_2^2)a_2^2\\
&\ +(3a_1^4b_2+15a_1^3b_2^2+2a_1^3+10a_1^2b_2+2a_1b_2^2-3a_1-b_2)a_2\\
&\ -4a_1^4+2a_1^3b_2-a_1^2-3a_1b_2-2.
\end{align*}

Thus $a_2=b_2$ if $f_3(a_1,a_2,b_2) \neq 0$. This is true because we
have some bounds for $a_1^2,a_1a_2,a_1b_2$ as presented in Lemma
\ref{bound} below, which can be applied to get
\begin{align*}
f_3(a_1,a_2,b_2)=&\ \ \ \ \ (20a_1^4b_2^2+15a_1^3b_2+3a_1^2b_2^2+2a_1b_2-4b_2^2)a_2^2\\
&\ +(3a_1^4b_2+15a_1^3b_2^2+2a_1^3+10a_1^2b_2+2a_1b_2^2-3a_1-b_2)a_2\\
&\ -4a_1^4+2a_1^3b_2-a_1^2-3a_1b_2-2\\
<&\ \ \ \ \ \frac{20}{5^4}+\frac{15}{5^3}+\frac{3}{4}a_2^2b_2^2+\frac{2}{5}a_2b_2-4a_2^2b_2^2\\
&\ +\frac{3}{2^2 5^2}+\frac{15}{5^3}+\frac{2}{2^2 5}+\frac{10}{5^2}+\frac{2}{5}a_2b_2-a_2b_2\\
&\ +\frac{2}{2^25}-2\\
<&\ -\frac{13}{4}a_2^2b_2^2-\frac{1}{5}a_2b_2-\frac{549}{500}\\
<&\ \ 0.
\end{align*}
Therefore, $f_3(a_1,a_2,b_2) \neq 0$ and $a_2=b_2$, just as
needed.\qedhere
\end{proof}

\begin{lem}\label{bound}\
\begin{enumerate} \item $a_1^2 \le \frac{1}{2}$,
\item $0 \le a_1 a_2 \le \frac{1}{5}$, and
\item $0 \le a_1 b_2 \le \frac{1}{5}$.
\end{enumerate}
\end{lem}
\begin{proof}(1) Let $A_i=1+a_1a_i$ for $i=1,\ldots,4$, then
$\sum \limits _{i=1}^4 A_i=4$, $A_1 \ge A_2 \ge 1$, $A_3 \ge A_4>0$
and
$$\frac{2}{A_1}+\frac{1}{A_2}+\frac{1}{A_3}+\frac{1}{A_4}=5.$$ Since
$$\frac{1}{A_3}+\frac{1}{A_4} \ge \frac{4}{A_3+A_4}=\frac{4}{4-A_1-A_2},$$
we have
\begin{equation} \label{A2} 5=\frac{2}{A_1}+\frac{1}{A_2}+\frac{1}{A_3}+\frac{1}{A_4} \ge
\frac{2}{A_1}+\frac{1}{A_2}+\frac{4}{4-A_1-A_2}.
\end{equation}
Let
$$g(A_2)=\frac{1}{A_2}+\frac{4}{4-A_1-A_2},$$ where $g$ is a
function in $\mathbb{R}[x]$. Then
$$\frac{\partial g(A_2)}{\partial A_2}=-\frac{1}{A_2^2}+\frac{4}{(4-A_1-A_2)^2}.$$
Note that $A_1 \ge A_2 \ge 1$, thus $4-A_1-A_2 \le 2$ and
$\frac{\partial g(A_2)}{\partial A_2} \ge 0$. Therefore $g(A_2) \ge
g(1)$ for $A_2 \ge 1$, that is,
$$\frac{1}{A_2}+\frac{4}{4-A_1-A_2} \ge 1+\frac{4}{3-A_1}.$$ Hence
by inequality (\ref{A2}),
$$5 \ge
\frac{2}{A_1}+\frac{1}{A_2}+\frac{4}{4-A_1-A_2} \ge
\frac{2}{A_1}+1+\frac{4}{3-A_1}.$$ We get $2A_1^2-5A_1+3 \le 0$,
i.e. $1 \le
A_1 \le \frac{3}{2}$. Thus $a_1^2 \le \frac{1}{2}$. \\

(2) Assume $A_2=1+a_1a_2 >\frac{6}{5}$. Then $g(A_2) >
g(\frac{6}{5})$. That is,
$$5 \ge
\frac{2}{A_1}+\frac{1}{A_2}+\frac{4}{4-A_1-A_2} >
\frac{2}{A_1}+\frac{5}{6}+\frac{4}{\frac{14}{5}-A_1}.$$ The solution
set of $A_1$ is $(-\infty,0)\cup (\frac{28}{25},\frac{6}{5}) \cup
(\frac{14}{5}, \infty)$. Note that $A_1>0$ and $A_1=1+a_1^2 \le
\frac{3}{2}$, we then get $\frac{28}{25}<A_1<\frac{6}{5}$, which
contradicts with $A_1 \ge A_2$.
Thus $A_2 \le \frac{6}{5}$ and $0 \le a_1a_2 \le \frac{1}{5}$.\\

(3) This result is followed by letting $A_1=1+a_1^2$ and
$A_i=1+a_1b_i$ for $i \ge 2$. The above proofs in part (1) and (2)
remain good.
\end{proof}

\begin{lem}\label{a3a4} $a_i=b_i$ for $i=3,4$.
\end{lem}
\begin{proof} Let $A_i=1+a_ib_1$ for $i=1,\ldots,4$. Then $$\sum \limits _{i=1}^4 A_i=4$$ and
$$\frac{2}{A_1}+\frac{1}{A_2}+\frac{1}{A_3}+\frac{1}{A_4}=5.$$
By the two equations above, since $A_3 \ge A_4$, we can derive
explicit expression for $A_3,A_4$ in the variables $A_1,A_2$, say
$A_3=h_1(A_1,A_2)$ and $A_4=h_2(A_1,A_2)$. If we let $B_i=1+a_1b_i$,
we can get $B_3=h_1(B_1,B_2)$ and $B_4=h_2(B_1,B_2)$ in a similar
manner. Note that $A_1=B_1$ and $A_2=1+a_2b_1=1+b_2a_1=B_2$, we
deduce that $A_i=B_i$ for $i=3,4$. Since $a_1=b_1>0$, $a_i=b_i$ for
$i=3,4$.
\end{proof}

By Lemmas \ref{a2} and \ref{a3a4}, we have $a_i=b_i$ for all $i$.
Hence $P$ can be expressed as
$$P=\begin{pmatrix}
    1+a_1^2 & 1+a_2 a_1 & 1+a_3 a_1 & 1+a_4 a_1 \\
    1+a_1 a_2 & 1+a_2^2 & 1+a_3 a_2 & 1+a_4 a_2 \\
    1+a_1 a_3 & 1+a_2 a_3 & 1+a_3^2 & 1+a_4 a_3 \\
    1+a_1 a_4 & 1+a_2 a_4 & 1+a_3 a_4 & 1+a_4^2 \\
  \end{pmatrix}
$$
where
\begin{equation} \label{zeroa}
\sum_{i=1}^4 a_i=0.
\end{equation}

By Corollary \ref{diffcor} we have the following system of equations

\begin{equation} \label{system}
\left\{ \begin{aligned} \frac{2}{1+a_1^2}+\frac{1}{1+a_2
a_1}+\frac{1}{1+a_3
a_1}+\frac{1}{1+a_4 a_1}&=5,\\
\frac{1}{1+a_1 a_2}+\frac{2}{1+a_2^2}+\frac{1}{1+a_3
a_2}+\frac{1}{1+a_4 a_2}&=5, \\
\frac{1}{1+a_1 a_3}+\frac{1}{1+a_2 a_3}+\frac{2}{1+a_3^2
}+\frac{1}{1+a_4 a_3}&=5, \\
\frac{1}{1+a_1 a_4}+\frac{1}{1+a_2 a_4}+\frac{1}{1+a_3
a_4}+\frac{2}{1+a_4^2}&=5.
\end{aligned} \right.
\end{equation}
With (\ref{zeroa}) and (\ref{system}), we claim that
\begin{lem} \label{final} $a_i=a_j$ if $a_i a_j>0$.
\end{lem}
\begin{proof} Let
\begin{equation}\nonumber \label{fx}
F(x)=\frac{1}{1+a_1 x}+\frac{1}{1+a_2 x}+\frac{1}{1+a_3
x}+\frac{1}{1+a_4 x}+\frac{1}{1+x^2}-5=0.
\end{equation}
Normalizing $F(x)$ yields a polynomial (the numerator) of degree 6
in $x$ whose constant is 0 and whose coefficient of the term $x$ is
$\sum \limits_{i=1}^4 a_i=0$. So $a_1,a_2,a_3,a_4,0,0$ are all the
zeros of $F(x)$. Suppose there exists consecutive $i,j$ such that
$a_i > a_j>0$ (or $a_j < a_i<0$ respectively). Then $F(x)$ is
continuous in the interval $(-\frac{1}{a_j},-\frac{1}{a_i})$. Note
that
$$\lim \limits_{x
\rightarrow -\frac{1}{a_j}^+}F(x)= \infty\ \textrm{ and } \lim
\limits_{x \rightarrow -\frac{1}{a_i}^-} F(x)= -\infty.$$
There must be a zero lying in $(-\frac{1}{a_j},-\frac{1}{a_i})$, say
$a_0$. Then $a_0 < -\frac{1}{a_i}$ (or $a_0 > -\frac{1}{a_j}$
respectively), i.e. $1+a_i a_0 <0$ (or $1+a_j a_0 <0$ respectively).
Since $a_0 \neq 0$, $x_{0}$ must be one of $a_k$, $k=1,\ldots,4$.
Thus $1+a_i a_0$ (or $1+a_j a_0$, respectively) is an entry in
matrix $P$, contradicting the fact that each entry of $P$ is
positive. Therefore if $i,j$ are consecutive and $a_i a_j>0$, we
must have $a_i=a_j$. Hence $a_i a_j>0$ implies $a_i=a_j$ for any
$i,j$.
\end{proof}

With Lemma \ref{final} it is handy to solve the system
(\ref{system}). Under Assumption (\ref{asp}) there are only 4
possible patterns of signs for $(a_1,a_2,a_3,a_4)$. If the signs are
$(+,+,+,-)$, then $a_1=a_2=a_3=-\frac{1}{3}a_4$. Substitute this to
any equation in (\ref{system}) yields
$a_1=a_2=a_3=\frac{1}{\sqrt{15}}$ and $a_4=-\frac{3}{\sqrt{15}}$.
The matrix would be
\begin{equation} \nonumber\label{m1}
P_1=\begin{pmatrix} \vspace{0.1cm}
  \frac{16}{15} & \frac{16}{15} & \frac{16}{15} & \frac{4}{5} \\
  \vspace{0.1cm}
  \frac{16}{15} & \frac{16}{15} & \frac{16}{15} & \frac{4}{5} \\
  \vspace{0.1cm}
  \frac{16}{15} & \frac{16}{15} & \frac{16}{15} & \frac{4}{5} \\
  \frac{4}{5} & \frac{4}{5} & \frac{4}{5} & \frac{8}{5}\\
\end{pmatrix}.
\end{equation}

For the case when the signs are $(+,+,-,-)$, we get
$a_1=\frac{1}{\sqrt{5}}$ and the matrix would be
\begin{equation} \label{max}
P_2=\begin{pmatrix} \nonumber\vspace{0.1cm}
  \frac{6}{5} & \frac{6}{5} & \frac{4}{5} & \frac{4}{5} \\
  \vspace{0.1cm}
  \frac{6}{5} & \frac{6}{5} & \frac{4}{5} & \frac{4}{5} \\
  \vspace{0.1cm}
  \frac{4}{5} & \frac{4}{5} & \frac{6}{5} & \frac{6}{5} \\
  \frac{4}{5} & \frac{4}{5} & \frac{6}{5} & \frac{6}{5} \\
\end{pmatrix}.
\end{equation}

When the signs are $(+,+,0,-)$, $a_1=\frac{1}{2\sqrt{2}}$, and the
matrix would be
\begin{equation} \nonumber\label{m2}
P_3=\begin{pmatrix}\vspace{0.1cm}
  \frac{9}{8} & \frac{9}{8} & 1 & \frac{3}{4} \\\vspace{0.1cm}
  \frac{9}{8} & \frac{9}{8} & 1 & \frac{3}{4} \\\vspace{0.1cm}
  1 & 1 & 1 & 1 \\
  \frac{3}{4} & \frac{3}{4} & 1 & \frac{3}{2} \\

\end{pmatrix}.
\end{equation}

And when the signs are $(+,0,0,-)$, $a_1=\frac{1}{\sqrt{3}}$ and the
matrix would be
\begin{equation} \nonumber\label{m3}
P_4=\begin{pmatrix}\vspace{0.1cm}
  \frac{4}{3} & 1 & 1 & \frac{2}{3} \\\vspace{0.1cm}
  1 & 1 & 1 & 1 \\\vspace{0.1cm}
  1 & 1 & 1 & 1 \\
  \frac{2}{3} & 1 & 1 & \frac{4}{3}\\
\end{pmatrix}.
\end{equation}

The matrices obtaining local maximum of the likelihood function must
be among the matrices above. We conclude that matrix $P_2$ obtains
the global maximum. Finally, multiplying matrix $P_2$ by
$\frac{1}{16}$ yields
\begin{equation}\nonumber
P= \frac{1}{40}
\begin{pmatrix}\vspace{0.1cm}
  3 & 3 & 2 & 2 \\\vspace{0.1cm}
  3 & 3 & 2 & 2 \\\vspace{0.1cm}
  2 & 2 & 3 & 3 \\
  2 & 2 & 3 & 3 \\
\end{pmatrix}.
\end{equation}

\subsection{Approach via Gr\"{o}bner bases}\label{comp}

Gr\"{o}bner basis technique is a general approach for solving systems of equations.
Buchberger introduced in 1965  the first algorithm for computing Gr\"{o}bner basis (see \cite{Buc85}),
 and subsequently there have been extensive efforts in improving its efficiency.  It is not  our purpose here to give a detailed survey of all the algorithms
 in the literature, but we mention two important algorithms F4 (Faug\`{e}re 1999, \cite{Fau99}) and F5 (Faug\`{e}re 2002, \cite{Fau02}) where signatures are introduced to detect useless S-pairs without performing reductions.  F5 is believed to be the fastest algorithm in the last decade.  Most recently, Gao,
Guan and Volny (2010, \cite{GGV10}) introduced an incremental algorithm (G2V) that is simpler and several times faster than F5,
and  Gao, Volny  and Wang (2010, \cite{GVW10}) developed a more general algorithm that avoids the  incremental nature of F5 and G2V and is flexible in signature orders.  All these algorithms are for general polynomial systems.
If a large system of polynomials  have certain structures, it is not known how to use these algorithms to take advantage of the structures of the polynomial system.

After we submitted our paper (in 2008),  one of the referees pointed out that  it is possible to compute the Gr\"{o}bner basis for our polynomial system
with $n=4$. We give more details on this computation.  The solution starts from Equations
(\ref{zero1}-\ref{diff1b}), using the scaling at of the beginning of Section\ref{handwr}. Without the scaling the solutions are infinite. For this one needs to assume $a_1=b_1 \ne 0$. Note that this assumption relies on Lemmas \ref{sign} and \ref{order} we proved. It takes about ten minutes for the whole computation in Maple on a moderate computer.

Precisely, one can construct an ideal
\begin{equation}\nonumber \mathcal{J}_0=\langle a_1-b_1, \sum_{i=1}^4 a_i, \sum_{i=1}^4 b_i, h_1,\cdots,h_8\rangle \subset \mathbb{C}[X]
\end{equation}
where $h_i$ is a numerator on the left hand side of Equations \ref{diff1a},\ref{diff1b}, $\mathbb{C}$ is the complex field  and $X$ represents the list of unknowns: $a_1,\cdots, a_4, b_1,\cdots, b_4$. Let
\[ \mathcal{J}_1=\mathcal{J}_0+\langle 1-u\cdot a_1\rangle \subset \mathbb{C}[X,u]\]
 where $u$ is a new variable. Then $a_1\ne 0$ for any solution of $\mathcal{J}_1$. We compute the Gr\"obner basis $G_1$ of $\mathcal{J}_1$
 in an elimination term order with $u>X$. Let $G_2 = G_1\cap \mathbb{C}[X]$. Then $G_2$ is a Gr\"obner basis of $\mathcal{J}_1 \cap \mathbb{C}[X]$.
 Now $\langle G_2 \rangle$ is a zero-dimensional ideal, and its rational univariate representation can be computed.  In this step, a univariate polynomial $r(v)$ with a new variable $v$ is computed, whose roots can represent all the solutions of $\langle G_2 \rangle$. It has degree of 398, with 56 real roots. By substituting each real root to the representations, there are 18 roots making that some entries of $P$ equal 0 thus $L(P)=0$. Each of the remaining solutions
 gives one of the following:
 \begin{align*}
P_1=\begin{pmatrix} \vspace{0.1cm}
  \frac{16}{15} & \frac{16}{15} & \frac{16}{15} & \frac{4}{5} \\
  \vspace{0.1cm}
  \frac{16}{15} & \frac{16}{15} & \frac{16}{15} & \frac{4}{5} \\
  \vspace{0.1cm}
  \frac{16}{15} & \frac{16}{15} & \frac{16}{15} & \frac{4}{5} \\
  \frac{4}{5} & \frac{4}{5} & \frac{4}{5} & \frac{8}{5}\\
\end{pmatrix},  \quad \quad
 P_2=\begin{pmatrix} \vspace{0.1cm}
  \frac{6}{5} & \frac{6}{5} & \frac{4}{5} & \frac{4}{5} \\
  \vspace{0.1cm}
  \frac{6}{5} & \frac{6}{5} & \frac{4}{5} & \frac{4}{5} \\
  \vspace{0.1cm}
  \frac{4}{5} & \frac{4}{5} & \frac{6}{5} & \frac{6}{5} \\
  \frac{4}{5} & \frac{4}{5} & \frac{6}{5} & \frac{6}{5} \\
\end{pmatrix}, \\
 P_3=\begin{pmatrix}\vspace{0.1cm}
  \frac{9}{8} & \frac{9}{8} & 1 & \frac{3}{4} \\\vspace{0.1cm}
  \frac{9}{8} & \frac{9}{8} & 1 & \frac{3}{4} \\\vspace{0.1cm}
  1 & 1 & 1 & 1 \\
  \frac{3}{4} & \frac{3}{4} & 1 & \frac{3}{2} \\
\end{pmatrix},   \quad \quad   \quad
P_4=\begin{pmatrix}\vspace{0.1cm}
  \frac{4}{3} & 1 & 1 & \frac{2}{3} \\\vspace{0.1cm}
  1 & 1 & 1 & 1 \\\vspace{0.1cm}
  1 & 1 & 1 & 1 \\
  \frac{2}{3} & 1 & 1 & \frac{4}{3}\\
\end{pmatrix},
\end{align*}
up to a permutation of variables $a_i$'s and $b_i$'s. It is straightforward to check that  $P_2$ is the optimal solution.

We also tried to the case for $n=5$, but our computation did not finish after more than one day,
mainly because the computation for the first Gr\"obner basis $G_1$ did not finish.
Gr\"{o}bner basis encodes both real and complex solutions. For our system with $n=4$, there are far more complex solutions than real solutions.
For a system of polynomials with finitely many complex solutions, it is expected that in general, the more solutions with  the system,
the harder to compute Gr\"{o}bner basis  (for any term order).  Also, even if a final Gr\"{o}basis is small,  the intermediate polynomials may be
large (in number of nonzero terms as well as the size of the coefficients), hence the algorithms can not finish in reasonable time in practice,
in fact, it's more likely that  the computer is out of memory quickly.
For our theoretical approach (by hand),  we were able to explore some partial structure
in our polynomial system. For example, we have a polynomial of the form $(a_2-b_2)f_3(a_1,a_2,b_2)$ in the
proof for Lemma \ref{a2}. Our approach is to justify  that the factor
$f_3(a_1,a_2,b_2)$, a trivariate polynomial with 17 terms, is
nonzero by applying some bounds from Lemma \ref{bound}, so that we
can derive the simplest equation $a_2-b_2=0$. In the proof we used
the fact that we are looking only for real solutions. However, it is
possible that $f_3(a_1,a_2,b_2)$ is zero for some complex solutions.
The locus of all solutions may be much more complicated than that of
real solutions, hence the Gr\"{o}bner basis is much more time
consuming to compute.

\section{Some comments on more general likelihood functions}\label{general}

In this section, we consider some generalization of the likelihood problem.  We  let
the exponent  in  the likelihood function (\ref{product}) be symbolic, and consider the function
\begin{equation}\label{newLP}
L(P)=\prod_{i=1}^n p_{ii}^s  \times \prod_{i \not= j}p_{ij}^t,
\end{equation}
where $P=(p_{ij})$ is still an $n \times n$ matrix as before. The question is how the optimal solution
depends on $(s,t)$. Even for the case when $n=4$, it seems hard to find the optimal solutions.
In the following, we describe some possible solutions in the form of conjectures.

\begin{cnj} For given $0<t<s$ where $t,s$ are two integers, among the set of all
non-negative $4 \times 4$ matrices whose rank is at most 2 and whose
entries sum to 1, the matrix
$$
P=\frac{1}{4s+12t}
\begin{pmatrix}
    \vspace{0.1in}
    \frac{s+t}{2} & \frac{s+t}{2} & t & t \\
    \vspace{0.1in}
    \frac{s+t}{2} & \frac{s+t}{2} & t & t \\
    \vspace{0.1in}
    t & t & \frac{s+t}{2} & \frac{s+t}{2} \\
    t & t & \frac{s+t}{2} & \frac{s+t}{2}
\end{pmatrix}$$
is a global maximum for the likelihood function L(P) in (\ref{newLP}) when $n=4$.
\end{cnj}

The results in Section \ref{genprop} remain good for this likelihood
function. The equation (\ref{diff1b}) becomes
\begin{equation*}
\frac{b_1}{1+a_i b_1}+\frac{b_2}{1+a_i b_2}+\frac{b_3}{1+a_i
b_3}+\frac{b_4}{1+a_i b_4}+\frac{(\frac{s}{t}-1) a_i}{1+a_i b_i}=0.
\end{equation*}
But the bounds in Lemma \ref{bound} involve the fraction
$\frac{s}{t}$ and become complicated. A similar equation to
(\ref{a2b2}) can be derived, but the nonzero factor is difficult to
claim. Hopefully we may also prove $a_2=b_2$. Then $a_3=b_3$ and
$a_4=b_4$ can be derived in a similar manner to Lemma \ref{a3a4}. So
does Lemma \ref{final}. Finally we can find 4 local extrema and need
only compare them to obtain the global maximum. In the case when the
signs of $(a_1,a_2,a_3,a_4)$ are $(+,+,+,-)$, we have the equation
$$a_1^2((3s+9t)a_1^2-(s-t))=0.$$ Thus
$a_1=\sqrt{\frac{s-t}{3s+9t}}$, and the matrix would be
\begin{equation*}
P_1=\begin{pmatrix}
  \vspace{0.1in}
  \frac{4s+8t}{3s+9t} & \frac{4s+8t}{3s+9t} & \frac{4s+8t}{3s+9t} & \frac{12t}{3s+9t} \\
  \vspace{0.1in}
  \frac{4s+8t}{3s+9t} & \frac{4s+8t}{3s+9t} & \frac{4s+8t}{3s+9t} & \frac{12t}{3s+9t} \\
  \vspace{0.1in}
  \frac{4s+8t}{3s+9t} & \frac{4s+8t}{3s+9t} & \frac{4s+8t}{3s+9t} & \frac{12t}{3s+9t} \\
  \frac{12t}{3s+9t} & \frac{12t}{3s+9t} & \frac{12t}{3s+9t} & \frac{12s}{3s+9t}\\
\end{pmatrix}.
\end{equation*}

In the case when the signs are $(+,+,-,-)$, we get
$a_1=\sqrt{\frac{s-t}{s+3t}}$ and the matrix would be
\begin{equation} \label{P2}
P_2=\begin{pmatrix}
  \vspace{0.1in}
  \frac{2s+2t}{s+3t} & \frac{2s+2t}{s+3t} & \frac{4t}{s+3t} & \frac{4t}{s+3t} \\
  \vspace{0.1in}
  \frac{2s+2t}{s+3t} & \frac{2s+2t}{s+3t} & \frac{4t}{s+3t} & \frac{4t}{s+3t} \\
  \vspace{0.1in}
  \frac{4t}{s+3t} & \frac{4t}{s+3t} & \frac{2s+2t}{s+3t} & \frac{2s+2t}{s+3t} \\
  \frac{4t}{s+3t} & \frac{4t}{s+3t} & \frac{2s+2t}{s+3t} & \frac{2s+2t}{s+3t} \\
\end{pmatrix}.
\end{equation}

One can prove that $L(P_1) < L(P_2)$ by some calculus technique, for
example, taking the partial derivative of $\frac{L(P_1)}{L(P_2)}$
with respect to $s$. In similar approaches one can also show that
$L(P_3) < L(P_2)$ and $L(P_4) < L(P_2)$ where $P_3, P_4$ are the
corresponding matrices for the cases when signs are $(+,+,0,-)$ and
$(+,0,0,-)$ respectively. Thus the matrix in (\ref{P2}) is a global
maximum.

More generally, let $(u)_{l_1\times l_2}$ be a block matrix with
every entry being $u$ where $l_1\times l_2 \in \mathbb{N}^2$ and
$u>0$.
\begin{cnj} \label{ts} Let $n \ge 2$ and $0<t<s$. Then the matrix
$$
P=\frac{1}{ns+(n-1)nt}
\begin{pmatrix}
\vspace{0.1in}
  (\frac{s-t}{\lceil \frac{n}{2}\rceil}+t)_{\lceil \frac{n}{2}\rceil \times \lceil \frac{n}{2}\rceil} &
  (t)_{\lceil \frac{n}{2}\rceil
  \times \lfloor \frac{n}{2}\rfloor} \\
  (t)_{\lfloor \frac{n}{2}\rfloor \times \lceil \frac{n}{2}\rceil} & (\frac{s-t}{\lfloor
  \frac{n}{2}\rfloor}+t) _{\lfloor \frac{n}{2}\rfloor \times \lfloor \frac{n}{2}\rfloor}
\end{pmatrix}$$
is a global maximum for L(P) in (\ref{newLP}).
\end{cnj}

\begin{cnj} Let $n \ge 2$ and $0<s \le t$.
Then the matrix
$$
P=
\begin{pmatrix}
\vspace{0.1in}
  \frac{2s}{n^2(s+t)} & \frac{1}{n^2} & \cdots & \frac{1}{n^2} & \frac{2t}{n^2(s+t)} \\
  \vspace{0.1in}
  \frac{1}{n^2} & \frac{1}{n^2} & \cdots & \frac{1}{n^2} & \frac{1}{n^2} \\
  \vspace{0.1in}
  \vdots & \vdots & \vdots & \vdots & \vdots\\
  \vspace{0.1in}
  \frac{1}{n^2} & \frac{1}{n^2} & \cdots & \frac{1}{n^2} & \frac{1}{n^2} \\
  \vspace{0.1in}
  \frac{2t}{n^2(s+t)} & \frac{1}{n^2} & \cdots & \frac{1}{n^2} & \frac{2s}{n^2(s+t)} \\
\end{pmatrix}$$
is a global maximum for L(P) in
(\ref{newLP}).\\
\end{cnj}

\subsection*{Acknowledgment} The authors were partially
supported by the National Science Foundation under grants DMS-0302549 and DMS-1005369
and National Security Agency under grant H98230-08-1-0030. We would like
to thank Bernd Sturmfels for his encouragement and anonymous
referees for their helpful comments, in particular one of them provided Maple codes to us.

\end{document}